\documentclass[letterpaper, 10 pt, conference]{ieeeconf}

\IEEEoverridecommandlockouts
\usepackage{graphics}
\usepackage{times} 
\usepackage{amsmath} 
\usepackage{amssymb} 
\usepackage{mathtools}
\usepackage{mathrsfs}
\usepackage{microtype}
\usepackage[noadjust]{cite}
\usepackage{xcolor}
\usepackage{bm}
\usepackage{tikz}
\usepackage{url}
\usepackage{siunitx}
\usepackage{booktabs}
\usepackage{multirow}
\usetikzlibrary{arrows.meta}

\title{\LARGE \bf
Data-Driven Power Flow for Radial Distribution Networks \\ with Sparse Real-Time Data}

\author{Oleksii Molodchyk, Omid Mokhtari, Samuel Chevalier, Mads R. Almassalkhi,
and Timm Faulwasser
\thanks{This research was (partially) funded in the course of TRR 391 \textit{Spatio-temporal
Statistics for the Transition of Energy and Transport} (520388526) by the Deutsche
Forschungsgemeinschaft (DFG).}
\thanks{O. Molodchyk and T. Faulwasser are with Institute of Control Systems, Hamburg
University of Technology, 21073 Hamburg, Germany.
{\tt\small oleksii.molodchyk@tuhh.de, timm.faulwasser@ieee.org}}
\thanks{O. Mokhtari, S. Chevalier, and M. Almassalkhi are with the Department of
Electrical and Biomedical Engineering, University of Vermont, Burlington, VT
05405, USA. {\tt\small \{omokhtar, schevali, malmassa\}@uvm.edu}}
}

\newcommand{\iunit}{\mathbf{j}}
\newtheorem{lemma}{Lemma}
\newtheorem{problem}{Problem}
\newtheorem{definition}{Definition}
\newtheorem{remark}{Remark}
\DeclareMathOperator{\vstack}{col}
\DeclareMathOperator{\rank}{rank}

\begin{document}
\maketitle
\thispagestyle{empty}
\pagestyle{empty}

\begin{abstract}
    Real-time control of distribution networks requires accurate information
    about the system state. In practice, however, such information is difficult
    to obtain because real-time measurements are available only at a limited
    number of locations. This paper proposes a novel data-driven power flow (DDPF) framework for balanced radial distribution networks. The proposed algorithm combines the behavioral approach with the DistFlow model and leverages offline historical data to solve power flow problems using only a limited set of real-time measurements. To design DDPF under sparse measurement conditions,
    we develop a sensor placement problem based on optimal network reductions. This allows us to
    determine sensor locations subject to a predefined sensor budget and to explicitly account for the radial nature of distribution networks. Unlike approaches
    that rely on full observability, the proposed framework is designed for
    practical distribution grids with sparse measurement availability.
    This enables data-driven power flow for real-time operation while reducing the number of required sensors. On several test cases, the proposed DDPF algorithm could demonstrate accurate voltage magnitude predictions, with a maximum error less than 0.001 p.u., with as little as 25\% of total locations equipped with sensors.
\end{abstract}

\section{INTRODUCTION} \label{sec:intro}
Model-based methods are widely used for monitoring, operation, and control
in power systems. These methods rely on explicit representations of network physics and operational constraints. Their use in real-time operation of distribution networks, however, becomes increasingly challenging with limited real-time data. This is in part due to the sheer size of medium- and low-voltage systems.
For example, according to the 2025 report~\cite[Table~22]{monitoringbericht2025}, German distribution networks
span close to two million kilometers in total length with thousands of electrical nodes. At the same time, operating conditions in such networks  can change rapidly due to fluctuations in load, distributed generation, and switching actions.
These rapid operating point variations require frequent updates of the network states for online monitoring. Even when an accurate network model is available, repeatedly solving model-based optimization or power flow  problems may be difficult to carry out in real-time, especially in large-scale distribution systems with nonlinear AC behavior.
This motivates the use of decomposition techniques and data-driven methods that recover system states directly from measured data. The latter is the focus of this work.

Among data-driven approaches for dynamical systems, the behavioral approach has recently received significant attention in the power-systems community \cite{markovskyDataDrivenControlBased2023, liDataBasedPredictiveControl2025, otzenDataDrivenOptimalPower2025}. In this framework, a system is characterized by the set of trajectories it can generate.
One of the most well known cases, in which this perspective is applicable is the case of linear time-invariant (LTI) systems. Specifically,~\cite{willemsNotePersistencyExcitation2005}
has shown that any valid input-output trajectory can be selected from a span of
finitely many pre-recorded trajectories of the same length.
The key advantage of this approach is that it bypasses
the intermediate step of system identification (or parameter estimation) potentially
making the adaptation of the controller to new data significantly faster.

Although the behavioral view allows one to quickly respond to new conditions by swapping the trajectory data, it comes with notable limitations. In particular, describing spaces of trajectories generated by general non-linear (static or dynamic) systems is non-trivial. Therefore, a lot of works on data-driven power flow focus on linear approximations allowing for a simplified analysis, e.g., \cite{molodchykDataDrivenMultiStageOPF2025}
uses ``DC'' power flow, while \cite{liDataBasedPredictiveControl2025}
utilizes LinDistFlow, a linear approximation centered around the DistFlow
equations \cite{baranOptimalCapacitorPlacement1989}. A work that stands out in
this context is \cite{otzenDataDrivenOptimalPower2025}, which draws
inspiration from \cite{jabrRadialDistributionLoad2006} and uses a conic
power flow model. However, to enable a trajectory-based description of
\cite{jabrRadialDistributionLoad2006}, the authors of \cite{otzenDataDrivenOptimalPower2025} impose  restrictive
assumptions: constant voltage magnitudes  for all nodes, voltage angles measured at every node, and reactive power flows are not accounted for.

In this work, we consider data-driven models for balanced, radial distribution networks.
As a first contribution, we extend the data-driven
approach from \cite{otzenDataDrivenOptimalPower2025} to the (non-linear)
DistFlow model, which allows the monitoring of voltage magnitudes, active, and reactive
power flows. We prove rigorous equivalence of our data-driven model with the radial Distflow model. However, this equivalence is only valid under complete measurement conditions (i.e., every state is measured). Thus, as a second contribution, we extend the DDPF methodology to networks with  sparse real-time data, i.e., we consider cases when the
outputs (i.e., voltages, power flows, etc.) cannot be measured at every location.
Leveraging \emph{historical}
data at every node, we select a suitable subset of sensor locations
considering a user-defined sensor budget. We then pass the \emph{online, real-time}
data measured only at the selected, budget-restricted locations to our proposed
data-driven DistFlow. Our numerical experiments suggest that for a practical range of sensor budgets, the solutions produced by the data-driven algorithm
remain meaningfully close to those obtained using full model information.

The rest of the paper is organized as follows. In Section~\ref{sec: prem}, we formulate the data-driven and DistFlow models and define the problem. Section~\ref{sec: dddf} develops the main methodology, including the data-driven DistFlow representation and the sensor placement framework. Section~\ref{sec: test case} presents numerical examples to illustrate the performance of the proposed approach. We conclude the paper in Section~\ref{sec: conclusion} where we discuss future research directions.

\section{PRELIMINARIES and PROBLEM STATEMENT} \label{sec: prem}
The trajectory-based view is well known to be particularly
attractive for discrete-time LTI systems of the form
\begin{equation}
    \label{eq:lti_system}
    \begin{array}{l l}
        \bm{x}(t+1) & = A \bm{x}(t) + B \bm{u}(t), \quad \bm{x}(0) \doteq x^0, \\
        \bm{y}(t)   & = C \bm{x}(t) + D \bm{u}(t).
    \end{array}
\end{equation}
Here, $x^{0}\in \mathbb{R}^{n_x}$ is some initial state, whereas the state $\bm{x}(t) \in \mathbb{R}^{n_x}$, input $\bm{u}(t) \in \mathbb{R}^{n_u}$, and output $\bm{y}(t) \in \mathbb{R}^{n_y}$ are coupled using real matrices
$A$, $B$, $C$, and $D$ for all $t \in \mathbb{N}_{0}$. We use the boldface symbols to distinguish trajectories from individual
signal values.

Traditionally, many works on behavioral theory target dynamical systems. However, in this paper, we focus on purely algebraic equations without coupling across time steps. Hence, as a special, \emph{state-less} case, we first consider a \emph{static system} \eqref{eq:lti_system}
with $n_{x}= 0$, i.e., reduced to a memory-less linear system of
equations
\begin{equation}
    \label{eq:lti_static}\bm{y}(t) = D \bm{u}(t).
\end{equation}

Representing an abstract, infinitely long trajectory
$\bm{w}: \mathbb{N}_{0}\to \mathbb{R}^{n_w}$ might not always be tractable. Hence,
for some $T \in \mathbb{N}$, we introduce the (vectorized) restriction $\bm{w}
    _{[0, T-1]}\doteq \vstack(\bm{w}(t))_{t=0}^{T-1}\in \mathbb{R}^{T n_w \times 1}$ of $\bm{w}$
to a $T$-long window $\{0, \ldots, T -1\} \subset \mathbb{N}_{0}$. Such
finite-length trajectories can be packaged into Hankel matrices.
\begin{definition}[Hankel matrix]
    Let $\bm{w}_{[0, T-1]}$ be a $T$-long trajectory. We define the Hankel
    matrix of depth $L \leq T$ as
    \[
        \mathscr{H}_{L}(\bm{w}_{[0, T-1]}) \doteq
        \begin{bmatrix}
            \bm{w}_{[0,L-1]} & \bm{w}_{[1,L]} & \cdots & \bm{w}_{[T-L,T-1]}
        \end{bmatrix},
    \]
    whose dimension is given by $L n_{w}\times (T-L+1)$. \hfill{\small{$\Box$}}
\end{definition}

The space of all possible fixed-length input-output trajectories of \eqref{eq:lti_system}
can be defined using Hankel matrices of a longer, pre-recorded trajectory of
\eqref{eq:lti_system}. This result is known in the systems and control
literature as the \emph{fundamental lemma}
\cite{willemsNotePersistencyExcitation2005, markovskyIdentifiabilityBehavioralSetting2023}. We state its restriction to static systems of form \eqref{eq:lti_static}.
\begin{lemma}[{\cite[Section~2]{willemsNotePersistencyExcitation2005}}]
    \label{lem:fl} Consider a pre-recorded input-output trajectory $(\bm{u}^{\mathrm{d}}
        _{[0,T-1]}, \bm{y}^{\mathrm{d}}_{[0,T-1]})$ generated by \eqref{eq:lti_static}. Let $L$ be
    a non-negative integer.
    If the pre-recorded data
    satisfies
    \[
        \rank
        \begin{bmatrix}
            \mathscr{H}_{L}(\bm{u}^{\mathrm{d}}_{[0,T-1]}) \\
            \mathscr{H}_{L}(\bm{y}^{\mathrm{d}}_{[0,T-1]})
        \end{bmatrix}
        = Ln_{u},
    \]
    then an $L$-long trajectory $(\bm{u}_{[0,L-1]}, \bm{y}_{[0,L-1]})$ can
    be generated by \eqref{eq:lti_system} if and only if
    \begin{equation}
        \label{eq:fl}
        \begin{bmatrix}
            \bm{u}_{[0,L-1]} \\
            \bm{y}_{[0,L-1]}
        \end{bmatrix}
        =
        \begin{bmatrix}
            \mathscr{H}_{L}(\bm{u}^{\mathrm{d}}_{[0,T-1]}) \\
            \mathscr{H}_{L}(\bm{y}^{\mathrm{d}}_{[0,T-1]})
        \end{bmatrix}
        g
    \end{equation}
    holds for some column selection vector $g \in \mathbb{R}^{T-L+1}$.
    \hfill {\small{$\Box$}}
\end{lemma}

Since the inputs and outputs of \eqref{eq:lti_static} are not coupled in time
it is enough to examine only unit-length trajectories ($L=1$) in the
context of Lemma~\ref{lem:fl}. Put differently, representing the behavior of \eqref{eq:lti_static}
boils down to collecting $n_{u}$ linearly independent input-output trajectories.
For example,
this approach has been applied to transmission networks by assuming data generated via a ``DC''
power flow model of the form \eqref{eq:lti_static} \cite{molodchykDataDrivenMultiStageOPF2025}.
Subsequently, we  discuss  distribution systems, which  require special care, especially if
detailed modeling including voltages is desired.

\subsection{Branch Flow Equations}
We consider a balanced radial AC power network at steady state, represented
by a tree-graph $(\mathcal{N}, \mathcal{E})$ with the set of nodes
$\mathcal{N}= \{0, 1, \ldots, n\}$ and edges
$\mathcal{E}\subset \mathcal{N}\times \mathcal{N}$. We fix the substation/slack
node $0 \in \mathcal{N}$ as the root and we use the shorthand $\mathcal{N}_{+}
    \doteq \mathcal{N}\setminus \{0\}$. The tree structure guarantees that each
node $i \in \mathcal{N}_{+}$ has a unique upstream parent
$j \in \mathcal{N}$ and a set of downstream children nodes. Each edge $(
    i,j)$ is directed such that $i$ points at its parent $j$, i.e., edges always point in the (upstream) direction towards the substation. Edges
are endowed with resistances $\texttt{r}_{ij}\geq 0$ and reactances $\texttt{x}_{ij}\geq 0$, resulting in impedances
$\lvert\texttt{z}_{ij}\rvert^{2}\doteq \texttt{r}_{ij}^{2}+ \texttt{x}_{ij}^{2}$ and $\lvert\texttt{z}_{ij}\rvert\neq 0$ for each $(i,j) \in \mathcal{E}$, which begets
the admittance matrix, $Y$.

Each $i \in \mathcal{N}$ is associated with the net apparent power injection
$p_{i}+ \iunit q_{i}\in \mathbb{C}$ and voltage phasor $V_{i}\in \mathbb{C}$.
For each $(i,j) \in \mathcal{E}$, we introduce the (sending-side) power flow
$P_{ij}+ \iunit Q_{ij}\in \mathbb{C}$. Throughout the paper, we make use of
the dependent variables $\ell_{ij}\in \mathbb{R}$ and
$\mathsf{v}_{i}\doteq \lvert V_{i}\rvert^{2}\in \mathbb{R}$ representing
the squared line current and nodal voltage magnitudes, respectively.

To describe the AC power flow, we introduce the DistFlow model
\cite{baranOptimalCapacitorPlacement1989} as the set of equations
\begin{subequations}
    \label{eq:df}
    \begin{alignat}
        {2}\label{eq:df_1} & P_{ij}= p_{i}+ \sum_{k: (k,i)\in \mathcal{E}}(P_{ki}- \texttt{r}_{ki}\ell_{ki}),                                         &  & \, \forall (i,j) \in \mathcal{E}, \\
        \label{eq:df_2}    & Q_{ij}= q_{i}+ \sum_{k: (k,i)\in \mathcal{E}}(Q_{ki}- \texttt{x}_{ki}\ell_{ki}),                                         &  & \, \forall (i,j) \in \mathcal{E}, \\
        \label{eq:df_3}    & \mathsf{v}_{j}= \mathsf{v}_{i}- 2(\texttt{r}_{ij}P_{ij}+\texttt{x}_{ij}Q_{ij})+\lvert\texttt{z}_{ij}\rvert^{2}\ell_{ij}, &  & \, \forall (i,j) \in \mathcal{E}, \\[8pt]
        \label{eq:df_4}    & \ell_{ij}\mathsf{v}_{i}= P_{ij}^{2}+Q_{ij}^{2},                                                                          &  & \, \forall (i,j) \in \mathcal{E}.
    \end{alignat}
\end{subequations}
Equations \eqref{eq:df_1} and \eqref{eq:df_2} describe the active and
reactive power balance, respectively; \eqref{eq:df_3} models the voltage drops
and \eqref{eq:df_4} ensures that each apparent power flow is equal to the
product of the respective voltage and current magnitudes.

It is well-known that DistFlow \eqref{eq:df} fully represents the power flows
of a single-phase equivalent model of a radial network~\cite{molzahnSurveyRelaxationsApproximations2019}.
As mentioned above, due to the nonlinearity of \eqref{eq:df_4} in the unknowns $\ell_{ij}$, $\mathsf{v}_i$, $P_{ij}$, and $Q_{ij}$ many papers
rely on linearized versions of DistFlow. In this paper, we establish a trajectory-based
equivalent of the nonlinear DistFlow \eqref{eq:df} akin to the model \eqref{eq:fl} for systems of
form \eqref{eq:lti_static}.

\subsection{Problem Statement}
\label{sub:prob} Since for each $i \in \mathcal{N}_{+}$, there exists only one $j \in \mathcal{N}$ such that $(i,j) \in\mathcal{E}$, we can abbreviate the edge-related quantities as
$P_{ij}\gets P_{i}$, $Q_{ij}\gets Q_{i}$ and $\ell_{ij}\gets \ell_{i}$. Now,
let the variables in \eqref{eq:df} be collected into (column) vectors
\begin{multline}
    \label{eq:measurements}p \doteq \vstack(p_{i})_{i \in \mathcal{N}_+},~q \doteq \vstack(
    q_{i})_{i \in \mathcal{N}_+},~P \doteq \vstack(P_{i})_{i \in \mathcal{N}_+}, \\
    Q \doteq \vstack(Q_{i})_{i \in \mathcal{N}_+},~ \ell \doteq \vstack(\ell_{i})_{i \in
    \mathcal{N}_+},~\mathsf{v}\doteq \vstack(\mathsf{v}_{i})_{i \in
    \mathcal{N}_+}.
\end{multline}
Applying the boldface notation for trajectories, we suppose that at time $t \in
    \mathbb{N}_{0}$, the system \eqref{eq:df} is excited with inputs
\begin{subequations}
    \label{eq:df_uy_data}
    \begin{equation}\label{eq:df_u_data}
        \bm{u}(t) \doteq \vstack(\bm{p}(t), \bm{q}(t)) \in \mathbb{R}^{2n}
    \end{equation}
    producing outputs
    \begin{equation}
        \label{eq:df_y_data}\bm{y}(t) \doteq \vstack(\bm{P}(t), \bm{Q}(t), \bm
            {\ell}(t), \bm{\mathsf{v}}(t), \bm{\mathsf{v}}_{0}(t)) \in \mathbb{R}
        ^{4n+1},
    \end{equation}
\end{subequations}
resulting in trajectories $\bm{u}: \mathbb{N}_{0}\to \mathbb{R}^{2n}$ and $\bm
    {y}: \mathbb{N}_{0}\to \mathbb{R}^{4n+1}$, respectively.

Note that the results available for LTI systems
do not readily extend to the nonlinear DistFlow \eqref{eq:df} due to the nonlinearity of \eqref{eq:df_4}. This
motivates the following problem statement.
\begin{problem}[Data-driven DistFlow]
\label{pr:dd_distflow} Let $\mathscr{B}_{\mathrm{DF}}\subset \mathbb{R}^{2n}
    \times \mathbb{R}^{4n + 1}$ be the manifold defined via
\[
    \mathscr{B}_\mathrm{DF} \doteq \left\{ \begin{array}{l}
        u = \vstack(p,q), \\
        y = \vstack(P,Q,\ell,\mathsf{v},\mathsf{v}_0)
    \end{array} \middle\vert ~\begin{array}{l}
        \text{DistFlow~\eqref{eq:df}} \\
        \text{satisfied}
    \end{array}  \right\}.
\]
Consider a pre-recorded, $T$-long input-output trajectory $(\bm{u}^{\mathrm{d}}_{[0,T-1]}, \bm{y}^{\mathrm{d}}
    _{[0,T-1]}) \in \mathscr{B}_{\mathrm{DF}} \times \dots \times \mathscr{B}_{\mathrm{DF}} \doteq \mathscr{B}_{\mathrm{DF}}^{T}$ structured according to
\eqref{eq:df_uy_data} and generated from \eqref{eq:df}. Determine
whether a candidate point
$(\tilde u,\tilde y) \in \mathbb{R}^{2n}\times \mathbb{R}^{4n + 1}$ belongs to
$\mathscr{B}_{\mathrm{DF}}$ based solely on the data
$(\bm{u}^{\mathrm{d}}_{[0,T-1]}, \bm{y}^{\mathrm{d}}_{[0,T-1]})$. \hfill
{\small{$\Box$}}
\end{problem}

Note that to record the output data $\bm{y}^{\mathrm{d}}_{[0,T-1]}$ in Problem~\ref{pr:dd_distflow}, in addition to measuring the slack voltage $\bm{\mathsf{v}}_0^{\mathrm{d}}(t)$, one has to place $n$ sensors across all nodes in $\mathcal{N}_+$ collectively measuring the 4-tuples
\begin{subequations}
    \begin{equation}\label{eq:y_measurements}
        (\bm{P}_i^{\mathrm{d}}(t), \bm{Q}_i^{\mathrm{d}}(t), \bm{\ell}_i^{\mathrm{d}}(t), \bm{\mathsf{v}}_i^{\mathrm{d}}(t)) \in \mathbb{R}^4, \quad \forall i \in \mathcal{N}_+,
    \end{equation}
    at every time $t \in \{0, \ldots, T-1\}$. Similarly, for each time point, the system inputs $\bm{u}^{\mathrm{d}}_{[0,T-1]}$ need to be collected via the tuples
    \begin{equation}\label{eq:u_measurements}
        (\bm{p}_i^{\mathrm{d}}(t), \bm{q}_i^{\mathrm{d}}(t)) \in \mathbb{R}^2, \quad \forall i \in \mathcal{N}_+.
    \end{equation}
\end{subequations}

As explained in Section~\ref{sec:intro}, the availability of \emph{real-time} measurements at every node of the network is challenging to achieve at distribution level. This leads to a problem formulation, which relaxes this assumption.

Consider a subset $\mathcal{R}\subset \mathcal{N}$ of the network nodes such that $\mathcal{R} \ni 0$ and let $\mathcal{R}_+ \doteq \mathcal{R} \setminus \{0\}$. Assume that sensors can only be hosted at nodes in $\mathcal{R}$. This means that as opposed to \eqref{eq:y_measurements}, the output vector now collects $\bm{\mathsf{v}}_0^{\mathrm{d}}(t)$ and
\begin{equation}
    (\bm{P}_i^{\mathrm{d}}(t), \bm{Q}_i^{\mathrm{d}}(t), \bm{\ell}_i^{\mathrm{d}}(t), \bm{\mathsf{v}}_i^{\mathrm{d}}(t)) \in \mathbb{R}^4, \quad \forall i \in \mathcal{R}_+,
\end{equation}
at each $t \in \{0, \ldots, T-1\}$. Thus, we denote the associated new, reduced output vector as
\[
    \bm{y}_\mathrm{r}^{\mathrm{d}}(t) \doteq \vstack(\bm{P}_\mathrm{r}^{\mathrm{d}}(t), \bm{Q}_\mathrm{r}^{\mathrm{d}}
    (t), \bm{\ell}_\mathrm{r}^{\mathrm{d}}(t), \bm{\mathsf{v}}_\mathrm{r}^{\mathrm{d}}(t), \mathsf{v}
    _{0}^{\mathrm{d}}(t)),
\]
where the subscript ``r'' denotes the restriction of \eqref{eq:measurements} to $\mathcal{R}_+$.

\begin{problem}[Relaxation to Sparse Real-Time Data] \label{pr:dd_distflow_reduced}
Given data
$(\bm{u}^{\mathrm{d}}_{[0,T-1]}, \bm{y}^{\mathrm{d}}_{\mathrm{r},[0,T-1]})$ with full inputs but reduced outputs. For an input vector (operating point) $\tilde u \in \mathbb{R}^{2n}$ of nodal power injections and some constant slack voltage,
\begin{itemize}
    \item[i)] Reconstruct the \textit{full} vector of voltage magnitudes $\lvert\tilde{V}(\tilde u)\rvert  = \vstack(\sqrt{\mathsf{v}_i})_{i \in \mathcal{N}_+} \in \mathbb{R}^n$ corresponding to $\tilde u$.
    \item[ii)] For a given \emph{sensor budget} $\beta^{\text{max}}$ and operating range $\mathbb{U} \subset \mathbb{R}^n \times \mathbb{R}^n$ of interest, find the set $\mathcal{R} \subset \mathcal{N}$ that solves the $\infty$-norm error problem
          \begin{equation}
              \begin{aligned}
                  \underset{\mathcal{R\subset\mathcal{N}}}{\text{minimize}} \quad & \max_{\tilde u \in \mathbb{U}}\lVert \lvert\tilde{V}(\tilde u)\rvert - \lvert V_\mathrm{true} (\tilde u)\rvert\rVert_\infty \\
                  \text{subject to:} \quad                                        & \lvert \mathcal{R} \rvert \leq \beta^{\text{max}},
              \end{aligned}
          \end{equation}
          where $\lvert V_\mathrm{true}(\tilde u)\rvert \in \mathbb{R}^n$ are the voltage magnitudes resulting from the \emph{full} DistFlow \eqref{eq:df} solved for $\tilde u$. \hfill {\small{$\Box$}}
\end{itemize}
\end{problem}

Note that reconstructing all voltage magnitudes from a dataset with reduced outputs $\bm{y}^{\mathrm{d}}_{\mathrm{r},[0,T-1]}$ is not possible, in general, without having the model $\eqref{eq:df}$. Hence, we propose a framework that begets a meaningful approximation.

\section{MAIN RESULTS -- DATA-DRIVEN DISTFLOW} \label{sec: dddf}
First, to tackle Problem~\ref{pr:dd_distflow} we start with a crucial observation: In the DistFlow equations
\eqref{eq:df}, the model topology (i.e., the underlying tree graph and
the associated line impedances) only affects the equations in the linear part \eqref{eq:df_1}--\eqref{eq:df_3}. The nonlinear equations  \eqref{eq:df_4} only depend on the unknowns and do not include any of the network parameters.

The following result proposes to learn the corresponding subspace using a rationale
similar to Lemma~\ref{lem:fl}.
\begin{lemma}[Data-driven DistFlow]
    \label{lem:dd_df} Consider an input-output trajectory $(\bm{u}^{\mathrm{d}}
        _{[0,T-1]}, \bm{y}^{\mathrm{d}}_{[0,T-1]}) \in \mathscr{B}_{\mathrm{DF}}^{T}$
    of length $T>1$, structured according to \eqref{eq:df_uy_data}
    and generated from \eqref{eq:df}. Assume that the trajectory satisfies
    \begin{equation}
        \label{eq:dd_df_pe}\rank
        \begin{bmatrix}
            \mathscr{H}_{1}(\bm{u}^{\mathrm{d}}_{[0,T-1]}) \\
            \mathscr{H}_{1}(\bm{y}^{\mathrm{d}}_{[0,T-1]})
        \end{bmatrix}
        = 3n + 1.
    \end{equation}
    Then, a candidate point
    $(u,y) \in \mathbb{R}^{2n}\times \mathbb{R}^{4n + 1}$ with
    $u=\vstack(p,q)$ and
    $y = \vstack(P, Q, \ell, \mathsf{v}, \mathsf{v}_{0})$ belongs to
    $\mathscr{B}_{\mathrm{DF}}$ if and only if there exists a vector
    $g \in \mathbb{R}^{T}$ such that
    \begin{subequations}
        \label{eq:dd_df}
        \begin{align}
            \label{eq:dd_df_lin}\begin{bmatrix}u \\ y\end{bmatrix} & = \begin{bmatrix}\mathscr{H}_{1}(\bm{u}^{\mathrm{d}}_{[0,T-1]}) \\ \mathscr{H}_{1}(\bm{y}^{\mathrm{d}}_{[0,T-1]})\end{bmatrix} g, \\
            \label{eq:con_cones}\mathsf{v}\odot \ell               & = P \odot P + Q \odot Q,
        \end{align}
    \end{subequations}
    where $\odot$ denotes the element-wise product. \hfill {\small{$\Box$}}

    \begin{proof}
        Note that \eqref{eq:con_cones} is a concatenation of the $n$
        equations in \eqref{eq:df_4}. Therefore, it is left to prove that the
        set $\{(u,y) \mid \exists g~\text{s.t.}~\eqref{eq:dd_df_lin}~\text{is
                satisfied}\}$ is equal to the set of the solutions to \eqref{eq:df_1}--\eqref{eq:df_3}.
        We express \eqref{eq:df_1}--\eqref{eq:df_3} as a linear system
        \[
            M
            \begin{bmatrix}
                u \\
                y
            \end{bmatrix}
            = 0~\text{with}~M \in \mathbb{R}^{3n \times (6n+1)}.
        \]
        Since $(\mathcal{N}, \mathcal{E})$ is a fully connected tree, we have
        $\rank M = 3n$. Therefore, by the fundamental theorem of linear algebra,
        the dimension of the nullspace of $M$ is $6n+1 - \rank M = 3n + 1$. Hence,
        the columns of the Hankel matrices in \eqref{eq:dd_df_pe} can serve as
        the basis for the subspace defined by \eqref{eq:df_1}--\eqref{eq:df_3}.
    \end{proof}
\end{lemma}

Note that the condition \eqref{eq:dd_df_pe} induces a necessary data length
requirement $T \geq 3n +1$ that scales linearly with the number of nodes in
the network. However, the sizes of the Hankel matrices in \eqref{eq:dd_df_lin}
scale quadratically with $n$ since Lemma~\ref{lem:dd_df} assumes that each of
the $T$ measurements is placed at every node.

\subsection{Convex Relaxation of Data-Driven DistFlow}
The prospect of the data-driven model \eqref{eq:dd_df} is that one can employ it efficiently
in optimization algorithms.
For instance, the data-driven solution to DistFlow-based power flow for a tuple of nodal net injections $(p_{\bullet}, q_{\bullet}) \in \mathbb{R}^{n}\times
    \mathbb{R}^{n}$ can be computed by obtaining a minimizer to
\begin{subequations}
    \label{eq:dd_df_pf}
    \begin{alignat}{2}
         & \underset{\substack{u = \vstack(p,q),\; g                                                                                                                                      \\ y = \vstack(P,Q,\ell,\mathsf{v},\mathsf{v}_0)}}{\text{minimize}}
         & \quad                                     & \mathbf{1}^{\top}\ell
        \\
         & ~~\,\text{subject to:}
         &                                           & \mathsf{v} \geq 0, \quad \mathsf{v}_{0} = 1~\text{p.u.}, ~\eqref{eq:dd_df_lin}
        \label{eq:dd_df_pf_v0}
        \\
         &                                           &                                                                                & \vstack( p, q)
        \leq
        \vstack(p_{\bullet}, q_{\bullet}),
        \label{eq:dd_df_pf_soc_os}
        \\
         &                                           &                                                                                & P \odot P + Q \odot Q \leq \mathsf{v} \odot \ell.
        \label{eq:dd_df_pf_soc}
    \end{alignat}
\end{subequations}
Here, we embed the linear part of the data driven model as equality
constraints, while relaxing the quadratic condition \eqref{eq:con_cones} to
\eqref{eq:dd_df_pf_soc}. Together with the voltage non-negativity
requirement \eqref{eq:dd_df_pf_v0} this casts \eqref{eq:dd_df_pf} into a
second-order cone program. Although this vastly improves the efficiency
of the optimization, one has to remark that the solution to \eqref{eq:dd_df_pf}
is only physically relevant when all of the conic constraints in \eqref{eq:dd_df_pf_soc}
are active, i.e., when the relaxation \eqref{eq:dd_df_pf_soc} is exact.

In the model-based setting, i.e., when \eqref{eq:dd_df_lin} is replaced by \eqref{eq:df_1}--\eqref{eq:df_3},
sufficient conditions to guarantee the relaxation exactness have been studied
extensively, see, e.g., \cite{farivarInverterVARControl2011, farivarBranchFlowModel2013,
    farivarBranchFlowModel2013a, ganExactConvexRelaxation2015}. Below we leverage
techniques from \cite{farivarInverterVARControl2011, farivarBranchFlowModel2013} to obtain
the data-driven counterpart to model-based convex-relaxed DistFlow. In particular,
observe that similar to \cite{farivarInverterVARControl2011} we employ load
oversatisfaction as the inequality constraint in \eqref{eq:dd_df_pf_soc_os}.
This allows us to derive the following result.
\begin{lemma}[Relaxation exactness]
    Suppose that the data in \eqref{eq:dd_df_pf} is noise-free and
    persistently exciting in the sense of \eqref{eq:dd_df_pe}. Then the
    optimal solution
    $(P^{\star}, Q^{\star}, \mathsf{v}^{\star}, \ell^{\star})$ to \eqref{eq:dd_df_pf}
    satisfies
    \[
        P^{\star}\odot P^{\star}+ Q^{\star}\odot Q^{\star}= \mathsf{v}^{\star}
        \odot \ell^{\star},
    \]
    i.e., the relaxation \eqref{eq:dd_df_pf_soc} is exact. \hfill {\small{$\Box$}}

    \begin{proof}
        Closely following \cite{farivarInverterVARControl2011, farivarBranchFlowModel2013}
        we can prove this statement by contradiction. Specifically, assume that
        $\zeta^{\star}= (p^{\star}, q^{\star}, P^{\star}, Q^{\star}, \mathsf{v}
            ^{\star}, \mathsf{v}_{0}^{\star}, \ell^{\star}, g^{\star})$ is optimal
        in \eqref{eq:dd_df_pf} but there exists a node
        $i$ whose cone constraint is inactive, i.e.,
        \[
            P_{i}^{\star2}+ Q_{i}^{\star2}< \mathsf{v}_{i}^{\star}\ell_{i}^{\star}
            .
        \]
        Let $j \in \mathcal{N}$ be the parent of $i$. As in \cite{farivarInverterVARControl2011}, pick some
        $\varepsilon > 0$ and consider a perturbed point
        $\tilde{\zeta}= (\tilde{p}, \tilde{q}, \tilde{P}, \tilde{Q}, \tilde{\mathsf{v}}
            , \tilde{\mathsf{v}}_{0}, \tilde{\ell}, \tilde{g})$
        defined by $\tilde{\mathsf{v}}\doteq \mathsf{v}^{\star}$ and $ \tilde{\mathsf{v}}_{0}\doteq \mathsf{v}_{0}^{\star}$, while for each $k \in \mathcal{N}_+$, set
        \begin{align*}
            (\tilde{\ell}_{k}, \tilde{P}_{k}, \tilde{Q}_{k}) & \doteq
            \begin{cases}
                (\ell_{k}^{\star}-\varepsilon,P_{k}^{\star}-\frac{\texttt{r}_{ij}\varepsilon}{2},Q_{k}^{\star}-\frac{\texttt{x}_{ij}\varepsilon}{2}),~\text{if}~k=i, \\
                (\ell_{k}^{\star},P_{k}^{\star},Q_{k}^{\star}),~\text{otherwise,}
            \end{cases}
            \\
            (\tilde{p}_{k},\tilde{q}_{k})                    & \doteq
            \begin{cases}
                (p^\star_{k}-\frac{\texttt{r}_{ij}\varepsilon}{2}, q^\star_{k}-\frac{\texttt{x}_{ij}\varepsilon}{2}),~\text{if $k \in \{i,j\}$}, \\
                (p^\star_{k}, q^\star_{k}),~\text{otherwise}.
            \end{cases}
        \end{align*}
        Since $\zeta^{\star}$ is feasible by definition, it satisfies the Hankel
        matrix constraints \eqref{eq:dd_df_lin} with some $g^{\star}$, and therefore
        it also satisfies \eqref{eq:df_1}-\eqref{eq:df_3} by Lemma~\ref{lem:dd_df}.
        Its perturbed version $\tilde{\zeta}$ also satisfies \eqref{eq:df_1}-\eqref{eq:df_3}.
        Using Lemma~\ref{lem:dd_df} in the other direction, we establish the
        existence of $\tilde{g}\neq g^{\star}$ such that $\tilde{\zeta}$
        satisfies \eqref{eq:dd_df_lin}. Since $\texttt{r}_{ij}\geq 0$ and
        $\texttt{x}_{ij}\geq 0$, $\tilde{\zeta}$ also fulfills the
        load oversatisfaction inequality \eqref{eq:dd_df_pf_soc_os}. Finally,
        $\tilde{\zeta}$ has a strictly smaller objective than
        $\zeta^{\star}$ since
        $\mathbf{1}^{\top}\ell^{\star}> \mathbf{1}^{\top}\tilde{\ell}$. Hence,
        the contradiction.
    \end{proof}
\end{lemma}

\begin{remark}[Role of $g$ and slack voltage]
    Notice that if the slack voltage trajectory dataset is constant, i.e.,
    if $\bm{\mathsf{v}_0}^{\mathrm{d}}(t) = 1$~p.u. for all $t$, then due to the fixed slack voltage in \eqref{eq:dd_df_pf_v0},
    the column selection decision variable $g$ is forced to satisfy
    $\mathbf{1}^{\top}g = 1$. This is similar to the treatment of LTI systems
    under constant disturbance \cite[Thm.~1]{berberichLinearTrackingMPC2022}.
    \hfill {\small{$\Box$}}
\end{remark}

Thus, the data-driven power flow in~\eqref{eq:dd_df_pf} can incorporate the non-linear DistFlow formulation for balanced, radial distribution networks. However, DDPF implementation will require that all nodes $\mathcal{N}$ are measured in real-time, which is not realistic. To enable DDPF, we adapt work on optimal network reductions to formulate a sensor placement problem that determines the subset of nodes $\mathcal{R}\subset \mathcal{N}$ for which online measurements can be provided.

\subsection{Determining Online Measurement Locations}
We formulate a sensor placement problem based on Kron reduction, which can be derived from Kirchhoff’s Current Law (KCL),~$I = Y V$.
Consider a partition of the admittance matrix $Y$ into a set of measured nodes $\mathcal{R}$, and a set of unmeasured nodes $\mathcal{U}$. We eliminate the current injections of the unmeasured nodes by \textit{assigning} them to the measured nodes $\mathcal{R}$, such that $I_\mathcal{U} = 0$. Then, KCL can be rewritten as
\begin{align} \label{eq:Y-partitioned}
    \left[\begin{array}{c}
                  I_\mathcal{R} \\
                  \hline
                  0
              \end{array}\right]
    =
    \left[\begin{array}{c|c}
                  Y_{\mathcal{R}\mathcal{R}} & Y_{\mathcal{R}\mathcal{U}} \\
                  \hline
                  Y_{\mathcal{U}\mathcal{R}} & Y_{\mathcal{U}\mathcal{U}}
              \end{array}\right]
    \left[\begin{array}{c}
                  V_\mathcal{R} \\
                  \hline
                  V_\mathcal{U}
              \end{array}\right].
\end{align}
Eliminating $V_\mathcal{U}$ from \eqref{eq:Y-partitioned} through substitution yields the Kron-reduced system $I_\mathcal{R} = Y_{\rm Kron}V_\mathcal{R}$ with
\begin{align} \label{eq:Y-kron}
    ~Y_{\rm Kron} \doteq Y_{\mathcal{R}\mathcal{R}} - Y_{\mathcal{R}\mathcal{U}} Y_{\mathcal{U}\mathcal{U}}^{-1} Y_{\mathcal{U}\mathcal{R}}.
\end{align}
Here, $Y_{\rm Kron}$ captures the equivalent network among the measured nodes that correspond to the sensor locations.
Accordingly, we model a network with a sparse measurement set as a reduced network, where each unmeasured (reduced) node is assigned to a measured (kept) representative. The approximation we employ here is adapted from optimal Kron-based network reduction (Opti-KRON)~\cite{chevalierOptimalKronbasedReduction2022,mokhtariStructurepreservingOptimalKronbased2025, optikron3}.

Let
$\Pi \in \{0,1\}^{(n+1) \times (n+1)}$ denote the representative assignment matrix,
where $\Pi_{i,i}= 1$ indicates that node $i$ is selected for measurement. If
node $j$ is not selected for measurement and its injection is diverted to a measured node $i$,
then $\Pi_{j,j}= 0$ and $\Pi_{i,j}= 1$. Additionally, $\Pi_{i,j}= 1$ indicates that $V_j$ is approximated by $V_i$.
We frame the placement of online measurements as the optimization problem
\begin{subequations}
    \label{eq: optikron}
    \begin{align}
        \underset{\Pi , V}{\text{minimize}}\quad &
        \max_t \ \lVert (\Pi^{\top}|V_{[t]}| - |\hat{V}_{[t]}|) \rVert_{\infty}\label{eq: obj-optk}                                                             \\
        \text{subject to:}\quad                  & Y V_{[t]} = \Pi \hat{I}_{[t]} \quad \forall t \in \{1,...,T\} \label{eq: kcl2},                              \\& \text{tr}(\Pi) \le \beta^{\text{max}},\label{eq: budget}
        \\
                                                 & \Pi^{\top}\mathbf{1}= \mathbf{1},\label{eq: bin1}                                                            \\
                                                 & \Pi \le \text{diag}(\Pi)\mathbf{1}^{\top}, \label{eq: bin2}                                                  \\
                                                 & \Pi_{i,j} \le \Pi_{i,k}, \quad \forall i,j \in \mathcal{N},\ \forall k \in \mathcal{V}_{i,j},\label{eq: adj} \\
                                                 & \Pi_{i,k} = 0, \quad \forall i,k \in \mathcal{N}: k \notin \mathcal{D}_i. \label{eq: downstream}
    \end{align}
\end{subequations}
Here,
$\hat{V}_{[t]} \in \mathbb{C}^{(n+1)}~\text{and}~\hat{I}_{[t]} \in \mathbb{C}^{(n+1)}$ denote the complex nodal voltages and current injections obtained from historic measurements at time $t$ (i.e., a historical data scenario).
The objective is to identify the optimal sensor locations and assignments that minimize the maximum voltage magnitude difference between the reduced and full networks.
Constraint~\eqref{eq: kcl2} enforces KCL after nodal aggregation to predict effect on reduced voltages. The
sensor budget is enforced by~\eqref{eq: budget} with
$\beta^{\text{max}}< n+1$. Constraint~\eqref{eq: bin1} ensures
that each node is assigned to exactly one measured node, while measured nodes cannot be assigned to other nodes, as enforced by~\eqref{eq: bin2}.
We ensure each cluster forms a connected sub-network using~\eqref{eq: adj}, where $\mathcal{V}_{i,j}$ denotes the set of internal nodes on a path connecting nodes $i$ and $j$. Specifically, if node $j$ is assigned to node $i$, then all the nodes in $\mathcal{V}_{i,j}$ must also be assigned to node $i$. Finally, equation~\eqref{eq: downstream} enforces that each selected measured node is the upstream node of its cluster, so that the upstream line-flow measurement is consistent with the aggregated cluster injection.
Here, the set $\mathcal{D}_i$ contains all downstream nodes of node $i$.

There are two practical challenges associated with~\eqref{eq: optikron}: scalability and radiality. These issues are discussed next.

\subsubsection*{Scalability and Optimality}
Note that~\eqref{eq: optikron} is a mixed-integer program (MIP) which scales poorly in general. To speed up determining the reduced network, we simplify the search-space to a sequence of single one-node reduction decisions and iteratively reduce the network until we satisfy the sensor budget, i.e., $\text{tr}({\Pi}) = \beta_{\max}$. This iterative implementation can explicitly compute every admissible single-node reduction \emph{in parallel} (across both nodes and scenarios) and rank the decisions, select the optimal single-node reduction based on~\eqref{eq: obj-optk}, and perform Kron reduction. The process then repeats one node at a time until sensor budget is reached. This algorithm is detailed in~\cite[Algorithm~1]{optikron3} and provides a feasible Kron-based reduction for~\eqref{eq: optikron}.

Upon completion, the iterative algorithm returns the assignment matrix $\Pi^\star$ and resulting Kron-reduced network:
\begin{equation} \label{eq:r_star}
    \mathcal{R}^\star \doteq \{i \in \mathcal{N} \mid \Pi^\star_{i,i} = 1\}.
\end{equation}

However, due to the nature of Kron reductions, the network realized by $\mathcal{R}^\star$ will be a (dense) meshed network\footnote{In particular, the Kron reduction of a radial graph results in an edge-disjoint union of maximal cliques, where any clique with more than three nodes introduces cycles.} for which DistFlow assumptions do not hold. To overcome this, we apply the \emph{radialization} from~\cite{mokhtariStructurepreservingOptimalKronbased2025} to recover an equivalent radial network from $\mathcal{R}^\star$.

\subsubsection*{Radialization}
The ``radialization'' procedure proposed in~\cite{mokhtariStructurepreservingOptimalKronbased2025}
determines the smallest set of previously reduced nodes $\mathcal{R}_\mathrm{rad}$  that have to be re-introduced to guarantee that the equivalent network $Y_{\rm Kron}$ is radial. We consider this overhead in the final radial network realized by
\begin{equation} \label{eq:radialization}
    \mathcal{R} \doteq \mathcal{R}^\star \cup \mathcal{R}_\mathrm{rad},
\end{equation}
which represents the solution to the sensor placement problem.
Figure~\ref{fig: clustering} demonstrates the sensor placement procedure,
from the original network topology to the clustered representation. Each cluster
hosts exactly one sensor located at node $i$ measuring $(\bm{P}_i^{\mathrm{d}}
    (t),\bm Q_{i}^{\mathrm{d}}(t),\bm \ell_{i}^{\mathrm{d}}(t),\bm{\mathsf{v}}_{i}
    ^{\mathrm{d}}(t))$ at time $t$. Thus, unmeasured nodes in a cluster (dashed) are represented by their measured proxy $i$ (solid).

We summarize the considered DDPF methodology in Fig.~\ref{fig:rectangles}:
in total, one can think of four representations---two for the reduced vs.
full topology case, and two more but stated from a data-driven, behavioral
approach perspective.
\begin{figure}
    \centering
    \includegraphics[width=0.9\linewidth]{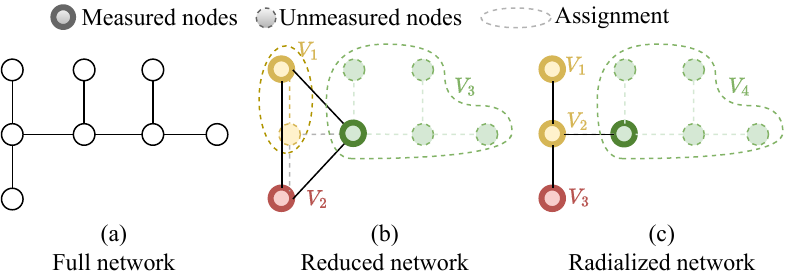}
    \caption{Different stages of the sensor placement process based on historical voltage and current data: (a) Original network
        topology. (b) Initial reduced network with three sensor locations. (c) Sensor locations
        after radialization.
    }
    \label{fig: clustering}
\end{figure}
\subsection{Data-Driven DistFlow with Reduced Measurements}
We arrive at the formulation of \eqref{eq:dd_df_pf} restricted to $\mathcal{R}$
\begin{alignat}{2} \label{eq:dd_pf_reduced}
     & \underset{\substack{u = \vstack(p,q),g, \ell_\mathrm{r}^\prime, \sigma_\ell,                                                                                                                                                                                                  \\ y_{\mathrm{r}}= \vstack(P_\mathrm{r},\ldots,\ell_\mathrm{r},\mathsf{v}_0)}}{\mathrm{minimize}}
     &                                                                              & \mathbf{1}^\top\ell^\prime_\mathrm{r}+f(P_\mathrm{r}, Q_\mathrm{r}) + \lambda_g \lVert g \rVert_2^2 + \lambda_\ell \lVert
    \sigma_\ell \rVert_2^2 \notag
    \\
     & ~~~~\text{subject to:}                                                       &                                                                                                                           & ~\text{constraints \eqref{eq:dd_df_pf_v0} and
    \eqref{eq:dd_df_pf_soc_os}}, \notag                                                                                                                                                                                                                                              \\
     &                                                                              &                                                                                                                           & ~
    \begin{bmatrix}
        u \\
        y_{\mathrm{r}}
    \end{bmatrix}
    =
    \begin{bmatrix}
        \mathscr{H}_{1}(\bm{u}^{\mathrm{d}}_{[0,T-1]}) \\
        \mathscr{H}_{1}(\bm{y}^{\mathrm{d}}_{\mathrm{r},[0,T-1]})
    \end{bmatrix}
    g,                                                                                                                                                                                                                                                                               \\
     &                                                                              &                                                                                                                           & ~\ell^\prime_\mathrm{r} = \ell_\mathrm{r} + \sigma_\ell, \notag    \\
     &                                                                              &                                                                                                                           & ~P_\mathrm{r} \odot P_\mathrm{r} + Q_\mathrm{r} \odot Q_\mathrm{r}
    \leq \mathsf{v}_\mathrm{r} \odot \ell_\mathrm{r}^\prime. \notag
\end{alignat}
In the above equation, besides reducing the outputs to $\mathcal{R}$, we introduce
a quadratic regularization on $g$ to avoid overfitting when selecting trajectories
from the image of the Hankel matrices. Additionally, since there is no counterpart of Lemma~\ref{lem:dd_df}
for the case of a limited sensor budget (see Fig.~\ref{fig:rectangles}), we introduce
a slack variable $\sigma_{\ell}$ in \eqref{eq:dd_pf_reduced} to avoid
running into infeasibility by adjusting the squared currents $\ell_\mathrm{r}$. To dis-incentivize
the load over-satisfaction introduced in \eqref{eq:dd_df_pf_soc_os}, we also modify
the objective by augmenting it with the regularization $
    f(P_\mathrm{r}, Q_\mathrm{r}) \doteq - \sum_{i: (i,0)\in \mathcal{E}}
    [P_{\mathrm{r},i} + Q_{\mathrm{r},i}]$
penalizing the power drawn from the slack node $0$.

\begin{figure}
    \centering
    \tikzstyle{every node}=[font=\small]
    \begin{tikzpicture}
        \definecolor{mylightblue}{RGB}{144, 213, 255}
        \definecolor{mylightgray}{RGB}{203,203,203}
        \def\w{2.5cm}
        \def\h{1.2cm}
        \def\gx{3cm}
        \def\gy{1cm}
        \fill[mylightblue] (0,\h+\gy) rectangle (\w, 2*\h+\gy);
        \node at (\w/2, \h+\gy+\h/2) {DistFlow~\eqref{eq:df}};
        \fill[mylightblue] (\w+\gx, \h+\gy) rectangle (2*\w+\gx, 2*\h+\gy);
        \node at
        (\w+\gx+\w/2, \h+\gy+\h/2)
        {\begin{tabular}{c}DistFlow on\\ reduced network\end{tabular}};
        \fill[mylightblue] (0,0) rectangle (\w, \h);
        \node at
        (\w/2, \h/2)
        {\begin{tabular}{c}Data-driven \\ DistFlow \eqref{eq:dd_df}\end{tabular}};
        \fill[mylightblue] (\w+\gx, 0) rectangle (2*\w+\gx, \h);
        \node at
        (\w+\gx+\w/2, \h/2)
        {\begin{tabular}{c} Data-driven \\ DistFlow on \\ reduced data\end{tabular}};
        \draw[
        mylightblue,
        line width=2pt,
        {Triangle[fill=mylightblue, scale=1]}-{Triangle[fill=mylightblue, scale=1]}
        ] (\w/4, \h+\gy) -- (\w/4, \h);
        \node[anchor=west] at (\w/3,\h+0.5\gy) {Lemma~\ref{lem:dd_df}};
        \draw[
        mylightblue,
        line width=2pt,
        -{Triangle[fill=mylightblue, scale=1]}
        ]
        (\w, \h+\gy+\h/2) -- (\w+\gx, \h+\gy+\h/2)
        node[midway, anchor=south, yshift=0.05cm]
            {\textcolor{black}{Optimization~\eqref{eq: optikron}}};
        \draw[
        mylightblue,
        line width=2pt,
        -{Triangle[fill=mylightblue, scale=1]}
        ]
        (\w, \h/2) -- (\w+\gx, \h/2)
        node[midway, anchor=south, yshift=0.05cm]
            {\textcolor{black}{\begin{tabular}{c}Sensor \\ placement\end{tabular}}};
    \end{tikzpicture}
    \caption{A summary of the considered DistFlow representations.}
    \label{fig:rectangles}
\end{figure}
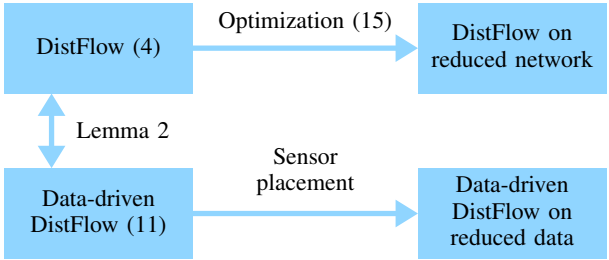

\section{NUMERICAL EXAMPLES} \label{sec: test case}
Although there is a lack of formal correspondence between the model-based
and data-driven DistFlow in the case of sparse real-time data,
we investigate whether the data-driven setting is able
to deliver reasonably accurate results in numerical experiments.

\subsubsection*{Setup and Test Cases}
We consider three radial networks with $47$, $85$, and $141$
nodes. The $47$-node case is taken from \cite{farivarInverterVARControl2011},
whereas the latter two networks are sourced from Matpower
\cite{Zimmerman2011}. The implementation is done in Julia and runs on a Intel Core i5-1335U processor with $16 \, \si{\giga \byte}$
RAM. The power flows in the model-based case with a given full
admittance matrix $Y$ (or a reduced admittance matrix $Y_\text{Kron}$),
are computed with a standard AC power flow feasibility problem and IPOPT \cite{Waechter2006}. For the conic problems in \eqref{eq:dd_df_pf} and
\eqref{eq:dd_pf_reduced} we use MOSEK \cite{mosek}.

For the considered networks, we use the German database in \cite{bdew2025slp}
to obtain standard profiles for three different types of nodes (household, commercial, and agricultural loads). Each node for each network
is randomly assigned to one of these three load categories.
This experiment design allows us to avoid having low-rank Hankel matrices.

The dataset is split into two parts: i) $T$-long training data for constructing
the Hankel matrices; and ii) test data of length $T_{\mathrm{test}}$ to provide
a series of operating points $(\bm{p}_{\bullet},\bm{q}_{\bullet})$ for DDPF based
on \eqref{eq:dd_pf_reduced}. The offline historical data to obtain network reductions
via \eqref{eq: optikron} is also taken from the training part of the dataset.

The training and test timeseries are both $24$ hours long with $15$-minute
time resolution.
The training set uses a workday in January (high-loading condition), and the test set a weekend day in July (low-loading condition).
It is important to note that for the proof-of-concept experiments presented here, all of the data points are taken as noise-free, i.e., we generate them by running AC power flow on the full network and we do not apply any further postprocessing.

\subsubsection*{Sensor Placement}
Table~\ref{tab: optikron_result} reports representative optimal sensor placement results for test cases in which the required measurement set achieves a maximum reconstruction error below $10^{-3}$ p.u. The reported number of sensors denotes the total number of measurements after radialization, while the added sensors correspond to the additional measurements introduced in \eqref{eq:radialization} after the initial optimization, cf.~Fig.~\ref{fig: clustering}. Our results suggest that the number of such additional measurements is less than $10\%$ of the total number of nodes and benefits from downstream constraint~\eqref{eq: downstream}, which appears to promote radiality. In practice, we recommend solving the sensor placement problem for multiple values of $\beta^{\max}$ and selecting a solution that satisfies the available sensing budget.

An example of the resulting sensor placement is shown in Fig.~\ref{fig: case141} for Case 141 from Table~\ref{tab: optikron_result}. The figure illustrates the selected sensor locations and their associated clusters in the original network, together with the corresponding Kron-reduced structures before and after radialization.
\begin{figure}
    \centering
    \includegraphics[width=\linewidth]{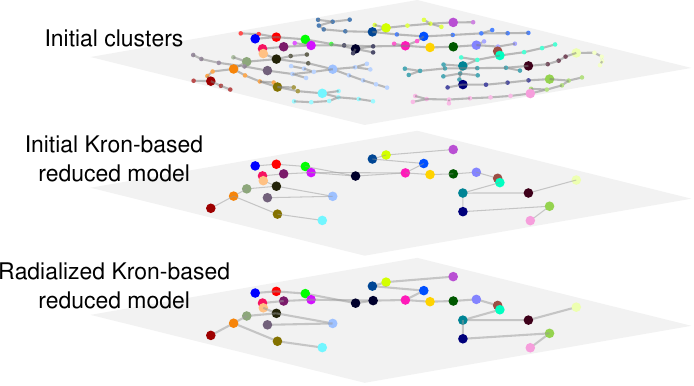}
    \caption{Sensor siting framework: selected sensor locations and their associated node clusters in the original network (top), Kron-based reduced network seen from the measured nodes in the original topology (middle), and radialized Kron-based reduced network seen from the measured nodes in the radialized topology (bottom).}
    \label{fig: case141}
\end{figure}
\begin{table}[t]
    \caption{Optimal sensor placement result examples}
    \label{tab: optikron_result}
    \centering
    \scriptsize
    \setlength{\tabcolsep}{4pt}
    \begin{tabular}{l|ccc}
        \toprule
        {Metric}                                                   & {Case 47}   & {Case 85}   & {Case 141}  \\
        \midrule
        {Reduction}                                                & {66.0\%}    & {54.1\%}    & {76.6\%}    \\
        \# of sensors $\lvert\mathcal{R}\rvert$                    & 16 (34.0\%) & 44 (45.9\%) & 33 (23.4\%) \\
        Radializing sensors $\lvert\mathcal{R}_\mathrm{rad}\rvert$ & 3           & 1           & 1           \\
        Largest cluster size                                       & 13          & 5           & 13          \\
        Max. Error \ (p.u.$\times 10^{-3}$)                        & 0.85        & 0.97        & 0.79        \\
        Computation time (s)                                       & 0.23        & 1.24        & 7.23        \\
        \bottomrule
    \end{tabular}
\end{table}

\subsubsection*{DDPF with Reduced Measurements}
Figure~\ref{fig:vm_errors} presents the maximum voltage magnitude prediction
error over all of the operation points in the test dataset. We sweep across different
sensor budgets, resulting in varying degrees of reduction represented in
percent via
\[
    \text{Reduction [\%]}= \frac{n+1 - \lvert\mathcal{R}\rvert}{n+1}\times 100.
\]
For each sensor budget setting, we solve the DDPF from \eqref{eq:dd_pf_reduced}
with regularization and slack parameters set to $\lambda_{g}\doteq 10^{-5}$
and $\lambda_{\ell}\doteq 10^{3}$, respectively. Comparing  this solution in terms of the computed voltage magnitudes to
running model-based power flow (full network with full admittance
matrix information $Y$) shows that  the errors remain in the
(acceptable) range of $[0,10^{-3}]$ p.u.; even for reductions of up to $40\%$
for all of three network cases. Additionally, it is worth mentioning that the times spent by MOSEK on solving the conic problem \eqref{eq:dd_pf_reduced} are negligible compared to the $15$-minute data sampling intervals.
For example, in the experiment from Figure~\ref{fig:vm_errors} for the $141$-node network, all 96 instances of \eqref{eq:dd_pf_reduced}
were solved within $100$ milliseconds for each degree of reduction.

To better understand how the errors are distributed across different nodes,
we examine one specific reduction instance. Figure~\ref{fig:case85_all_nodes}
displays the voltage magnitude errors recovered from a $60\%$-reduced 85-node
grid. Notably, DDPF appears to deliver biased estimations: for nearly all nodes
the voltages are consistently overestimated.
\begin{figure}
    \centering
    \includegraphics[width=\linewidth]{
        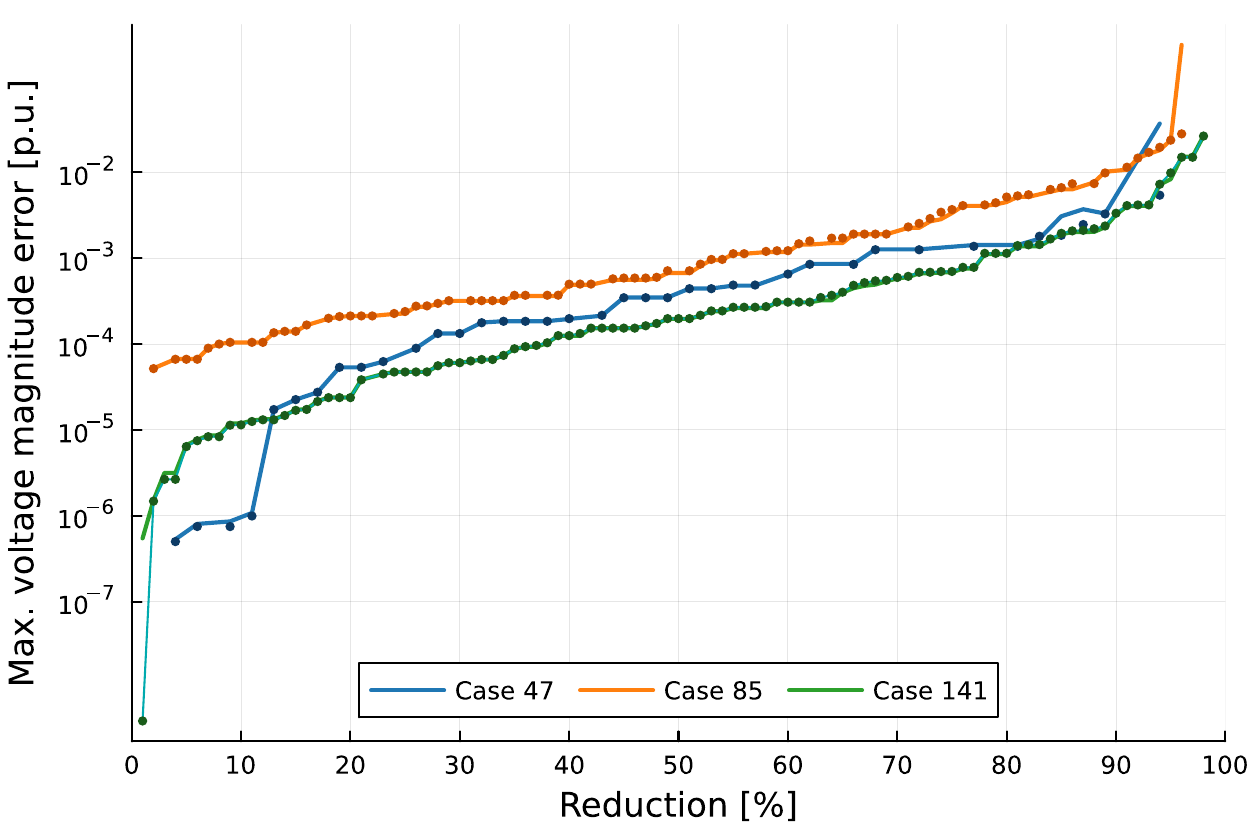
    }
    \caption{Maximum voltage magnitude errors over the test dataset and over
        varying degrees of reduction (solid lines represent the error resulting from DDPF,
        whereas darker-shaded bullets depict errors from model-based power flow
        calculations on networks with reduced models resulting from \eqref{eq: optikron} and subsequent radialization).}
    \label{fig:vm_errors}
\end{figure}

\begin{figure}
    \centering
    \includegraphics[width=\linewidth]{
        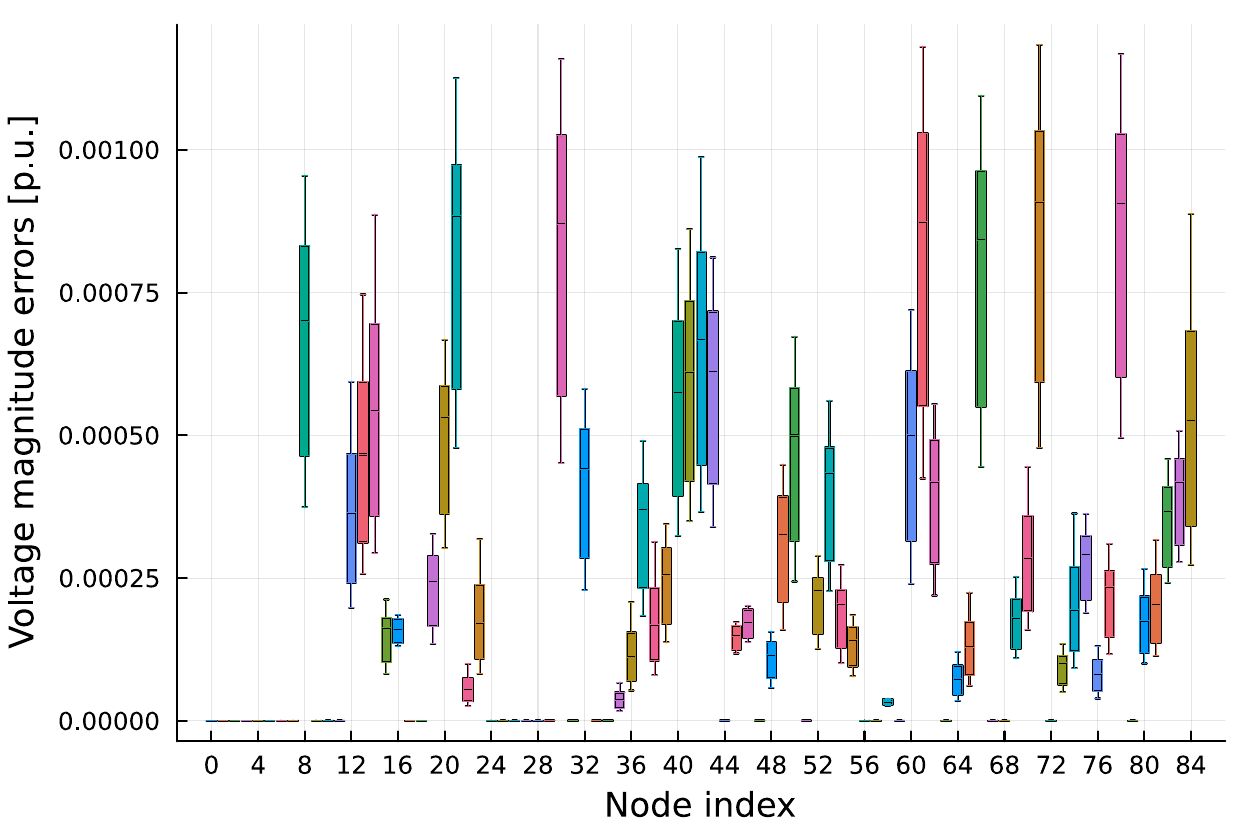
    }
    \caption{Voltage magnitude errors of DDPF in terms of differences to the
        model-based power flow using full topology information. The plot displays
        error statistics over all test set data points for a $60\%$ reduction instance
        of the 85-node network and across all of its 85 nodes.}
    \label{fig:case85_all_nodes}
\end{figure}

\section{CONCLUSIONS AND FUTURE WORK} \label{sec: conclusion}
This paper presented a novel framework for data-driven power flow based on
trajectories of the nonlinear DistFlow model. We have introduced and analyzed an exact data-driven surrogate of the DistFlow model. We have also shown how this model can be combined with known conic relaxation techniques.
To avoid depending on measurements at every node, we also investigate how to obtain a Kron-reduced radial model.
In particular, we addressed the problem of limited
sensor budgets through the offline solution of an
optimal network reduction problem that determines the subset of nodes to
be equipped with sensors.

The preliminary results herein are promising and offer numerous future research directions. In particular, extending the DDPF method to account for measurement noise and network reduction errors, adapting it to unbalanced distribution feeders, and leveraging the proposed framework for real-time control of distributed energy resources under sparse observability are of immediate interest.
Another direction is to assess the robustness of DDPF and Opti-KRON under different unseen loading conditions.

\section{ACKNOWLEDGMENTS}
The authors would like to acknowledge the Fourth Champéry Power Conference as the venue where the initial idea for this work was conceived.

\bibliographystyle{ieeetr}
\bibliography{literature}
\end{document}